\documentclass[11pt]{article}
\usepackage[utf8]{inputenc}
\usepackage[T1]{fontenc}
\usepackage[a4paper, margin=25mm]{geometry}
\usepackage{lmodern,microtype,algorithm,algpseudocode,amsthm,amsfonts,mathtools,thmtools,thm-restate,tikz,graphicx,hyperref}
\usepackage{enumitem}\setlist[enumerate,1]{label={(\alph*)}}

\usetikzlibrary{arrows.meta}

\theoremstyle{plain}
\declaretheorem[name=Theorem]{theorem}

\declaretheorem[sibling=theorem,name=Corollary]{corollary}

\declaretheorem[name=Claim]{claim}

\newcommand{\NP}{\text{NP}}
\newcommand{\DP}{\text{P}}

\newcommand{\cogem}{\text{co-gem}}
\newcommand{\cogems}{\text{co-gems}}

\newcommand{\miniaturePainting}{\textsc{Miniature Painting}}

\newcommand{\freeFloodIt}{\textsc{Free Flood-It}}

\newcommand{\area}{painting area}

\newcommand{\areas}{painting areas}
\newcommand{\recursive}{recursive}

\newcommand{\plan}{painting plan}

\newcommand{\plans}{painting plans}

\title{Fanciful Figurines flip Free Flood-It -- Polynomial-Time Miniature Painting on Co-gem-free Graphs}
\author{Christian Rosenke\\[0.5em]University of Rostock\\Germany\\[0.2cm]\texttt{\href{mailto:christian.rosenke@uni-rostock.de}{christian.rosenke@uni-rostock.de}}
	\and Mark Scheibner\\[0.5em]Independent Researcher\\Germany\\[0.2cm]\texttt{\href{mailto:mail@mark-scheibner.de}{mail@mark-scheibner.de}}}
\date{}

\begin{document}
	\maketitle
	\begin{abstract}
		\noindent
		Inspired by the eponymous hobby, we introduce \textsc{Miniature Painting} as the computational problem to paint a given graph $G=(V,E)$ according to a prescribed template $t \colon V \rightarrow C$, which assigns colors $C$ to the vertices of $G$. 
		In this setting, the goal is to realize the template using a shortest possible sequence of brush strokes, where each stroke overwrites a connected vertex subset with a color in $C$. 
		
		We show that this problem is equivalent to a reversal of the well-studied \textsc{Free Flood-It} game, in which a colored graph is decolored into a single color using as few moves as possible. 
		This equivalence allows known complexity results for \textsc{Free Flood-It} to be transferred directly to \textsc{Miniature Painting}, including $\NP$-hardness under severe structural restrictions, such as when $G$ is a grid, a tree, or a split graph. 

		Our main contribution is a polynomial-time algorithm for \textsc{Miniature Painting} on graphs that are free of induced {\cogems}, a graph class that strictly generalizes cographs. 
        As a direct consequence, \textsc{Free Flood-It} is also polynomial-time solvable on {\cogem}-free graphs, independent of the initial coloring. 
	\end{abstract} 
	
%

\section{Introduction}
\label{sec:introduction}

\emph{Miniature painting} is a popular hobby in which small metal or plastic model figures are painted in order to represent a wide variety of heroes and monsters.
The figures are three-dimensional objects of varying complexity whose outer surfaces consist of many small two-dimensional regions.
Each of them is the subject of a meticulous {\plan}, executed with the utmost precision.
In mathematics, it is quite customary to strip something wonderful of everything that sparks joy by placing it in a formal framework.
We now proceed accordingly.
\begin{figure}[!ht]
\tikzstyle{cntrst} = [draw=white, line width=0.6pt, double=black, double distance=0.6pt]
\tikzstyle{cntrstLabel} = [{Circle[length=3pt]}-, cntrst, rounded corners]
\tikzstyle{rightText} = [right = 0.3cm, font=\scriptsize]
\tikzstyle{simple} = [circle, minimum size=6mm, inner sep=0pt, fill = white, font=\small, draw = black]
\tikzstyle{vertex} = [simple, cntrst]
\tikzstyle{edge} = [line width = 0.6pt]
\tikzstyle{gr} = [fill = green!75!white, circle, minimum size=10mm, inner sep=0pt]
\tikzstyle{br} = [fill = brown!50!brown, circle, minimum size=10mm, inner sep=0pt]
\tikzstyle{lb} = [fill = brown!50!white, circle, minimum size=10mm, inner sep=0pt]
\tikzstyle{mt} = [fill = black!50!white, circle, minimum size=10mm, inner sep=0pt]
\tikzstyle{wh} = [fill = white!90!black, circle, minimum size=10mm, inner sep=0pt]
\tikzstyle{bl} = [fill = black!75!white, circle, minimum size=10mm, inner sep=0pt]
\begin{center}
\begin{tikzpicture}[baseline=1cm,yscale=0.8]
	\begin{scope}[rotate=90, scale=2]
	\path 	(0,0)	node[gr]{} node[vertex] (b)	{$b$}	
			(0,-1) 	node[br]{} node[vertex] (p)	{$p$}	
			(0,-2) 	node[gr]{} node[vertex] (f)	{$f$}	
			(1,-1) 	node[lb]{} node[vertex] (s)	{$s$}	
			(1,-2) 	node[wh]{} node[vertex] (m)	{$m$}	
			(-1,-2) node[mt]{} node[vertex] (c)	{$c$}	
			(-1,-1) node[bl]{} node[vertex] (l)	{$l$}	
			(-1,0)  node[wh]{} node[vertex] (x)	{$x$}	
			(-1,1)  node[wh]{} node[vertex] (h)	{$h$}	
			(0,1)   node[wh]{} node[vertex] (y)	{$y$}	
			(1,1)   node[bl]{} node[vertex] (r)	{$r$}	
			(1,0)   node[lb]{} node[vertex] (g)	{$g$};	
	\foreach \n in {p,s,c,l,x,h,y,r,g}{\draw[edge] (b) -- (\n);}
	\foreach \x/\y in {p/f,p/s,s/f,s/m,f/m}{\draw[edge] (\x) -- (\y);}
	\end{scope}
	\node at (1, -3.25) {(a)};
\end{tikzpicture}
\hspace*{1.5cm}
\begin{minipage}[t]{0.3\textwidth}
\begin{center}
\begin{tikzpicture}[baseline=-0.5cm]
\path 	(0,0)	node[simple] (v1)	{$v_1$}
		(1,0) 	node[simple] (v2)	{$v_2$}
		(2,0) 	node[simple] (v3)	{$v_3$}
		(3,0) 	node[simple] (v4)	{$v_4$};
	\foreach \x/\y in {v1/v2,v2/v3,v3/v4}{\draw[edge] (\x) -- (\y);}
	\node at (1.5, -1) {(b)};
\end{tikzpicture}

\vspace*{1cm}
\begin{tikzpicture}[baseline=0cm]
\path 	(0,0)	node[simple] (v1)	{$v_1$}
		(1,0) 	node[simple] (v2)	{$v_2$}
		(2,0) 	node[simple] (v3)	{$v_3$}
		(3,0) 	node[simple] (v4)	{$v_4$}
		(4,0) node[simple] (x)	{$x$};
	\foreach \x/\y in {v1/v2,v2/v3,v3/v4}{\draw[edge] (\x) -- (\y);}
	\node at (2, -1) {(c)};
\end{tikzpicture}

\end{center}
\end{minipage}
\end{center}
\caption{(a)~a graph model with template, (b)~the path $P_4$ on four vertices, (c)~the {\cogem}}
\label{fig:example}
\end{figure}
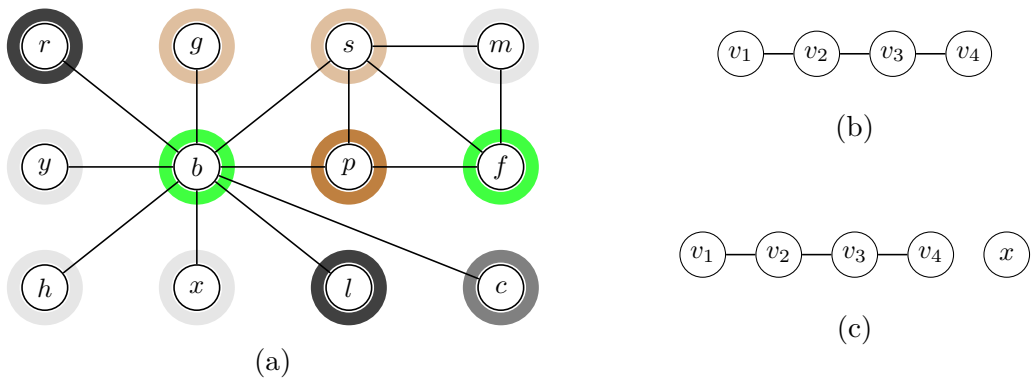

Any figurine can, somewhat drily, be represented as a graph $G=(V,E)$, with vertex set $V$ corresponding to all the surface patches, and edges $vw \in E$ indicating adjacency between patches $v$ and $w$.
The intended painting can then be captured by a \emph{template} function $t\colon V \rightarrow C$, which assigns to each vertex $v \in V$ the color $t(v) \in C$ chosen by the hobbyist for that patch.
Here, $C$ denotes the set of all available colors.
Figure~\ref{fig:example}~(a) depicts an example graph along with the laid-out painting that uses the color set $C = \{B, G, K, L, M, W\}$ to represent the six colors $B$rown, $G$reen, blac$K$, $L$ight brown, $M$etal, and $W$hite.

While the miniature painter may find fulfillment in the act of painting, the theoretical computer scientist finds it in algorithmic optimization.
Although the setting described above admits many useful optimization objectives, this paper focuses on the somewhat unconventional goal of minimizing the number of brush strokes.
The idea is to spare the painter from lifting the brush off the miniature more often than necessary.
Assuming the brush can hold a sufficient amount of paint, a single stroke is defined as the act of placing the brush onto the figurine and then moving it progressively across adjacent patches without lifting the brush tip, thereby filling all visited patches with the currently loaded color.

In the graph setting, a \emph{brush stroke} is, thus, modeled as the selection of a connected vertex set $A \subseteq V$ of $G$ and the assignment of a color $c \in C$ to all vertices of the \emph{\area} $A$.
The problem is then to compute an optimal \emph{{\plan}}, a shortest sequence $(A_1, c_1), \dots, (A_s, c_s)$ of strokes that produces the target template $t$ starting from an initially unpainted graph $G$.
The example graph $G$ of Figure~\ref{fig:example}~(a) can be painted with the given template using for instance the following plan of six strokes:
\[(V(G), W), (\{b,f,g,l,p,r,s\}, K), (\{b,f,g,p,s\}, L), (\{b,f,p\}, G), (p, B), (c,M).\]
Clearly, this plan is optimal since every used color requires at least one stroke, a lower bound that is met here.

The abstract problem outlined above probably admits a range of meaningful practical applications.
As a somewhat whimsical example, one might imagine an automated painter operating in a \emph{color-by-numbers} setting, which would benefit from an optimal {\plan} if each brush lift entails a costly and error-prone recalibration of its motion system.
In truth, however, our motivation is simply the intrinsic appeal of the underlying formal problem.

\medskip
To the best of our knowledge, this specific problem has not previously been studied in the form considered here.
Consequently, its computational complexity appears to be open, in particular under the many natural restrictions that may be imposed on the input.
Somewhat surprisingly, however, much of this question is already resolved implicitly.
Indeed, the {\miniaturePainting} problem studied in this work turns out to be equivalent to {\freeFloodIt}, which has been investigated extensively in the literature and is surveyed in several excellent overview papers~\cite{Fellows2018,Fleischer2010}.
This equivalence is unexpected because, while our objective is to colorfully \emph{paint} a graph using as few strokes as possible, {\freeFloodIt} is a game in which the player aims to \emph{decolor} a graph using as few moves as possible.
Nevertheless, as unintuitive as this correspondence may seem, the equivalence of the two problems is established in Section~\ref{sec:equivalence}.
More precisely, we show that a graph $G$ can be painted according to a template $t$ in $s$ strokes, if and only if the player can flood the graph $G$ monochromatically in $s-1$ moves, when $G$ is initially colored according to $t$.
For instance, the {\plan} for the example in Figure~\ref{fig:example}~(a) can be \emph{inverted} into the following five-move free-flooding of $G$ when initially painted with the given template:
\[(c,W), (p,G), (p,L), (p,K), (p,W).\]

A direct consequence of this equivalence is that all known complexity results for {\freeFloodIt} immediately apply to {\miniaturePainting} as well.
In particular, the problem remains computationally hard even under very strong input assumptions, for example when the input graph $G$ is restricted to \emph{co-comparability}, \emph{AT-free}, or \emph{split graphs}~\cite{Fukui2013}.
Even if $G$ is a \emph{tree}~\cite{FELLOWS2015} or a \emph{grid}~\cite{Clifford2012} and the template uses only three colors, no efficient algorithm exists unless $\DP=\NP$.
Indeed, there are very few non-trivial graph classes for which {\freeFloodIt} -- and hence {\miniaturePainting} -- is known to be tractable.
Basically, this includes just very restricted grids~\cite{MEEKS2012} and paths and cycles~\cite{MeeksScott2014}.
Moreover, it is folklore that the problem is tractable on \emph{cographs}, which arise from forbidding induced $P_4$, the path on four vertices shown in Figure~\ref{fig:example}~(b).
To the best of our knowledge, however, this has never been proven explicitly.

\medskip
As one of the few known outposts along the boundary between tractability and assumed computational hardness, cographs provide a natural starting point for extending efficient algorithms to {\miniaturePainting}.
Our aim is therefore to reintroduce induced $P_4$ into the input graph and to explore how far this extension can be carried without crossing into $\NP$-hard territory.
Accordingly, the main contribution of this paper is a polynomial-time algorithm for {\miniaturePainting} on graphs that are free of induced {\cogems}, the graph that extends the $P_4$ with one independent vertex as illustrated in Figure~\ref{fig:example}~(c).
This result, presented in Section~\ref{sec:polynomial_time_algorithm}, constitutes a substantial generalization.
The class of {\cogem}-free graphs strictly contains all cographs and also encompasses a wide range of graphs that admit induced $P_4$.
In fact, it is a simple but crucial observation that this graph class allows induced $P_4$ precisely when every such $P_4$ forms a \emph{dominating set}, that is, a vertex subset whose vertices collectively are adjacent to all other vertices in $V(G)$.
Our polynomial-time approach is fundamentally based on exploiting this structural property.
The example graph of Figure~\ref{fig:example}~(a) contains induced $P_4$ such as $r,b,s,m$, so it is not a cograph, but observe that there is no induced {\cogem}.

The equivalence established above transfers this progress directly back to {\freeFloodIt}: an optimal flooding strategy can be computed efficiently on graphs that are free of induced {\cogems}, independently of the complexity of the initial coloring.
This naturally raises the question of why one should introduce the {\miniaturePainting} problem at all, rather than studying {\freeFloodIt} directly.
Our answer is that the new formulation provides a complementary perspective on the game.
In the {\miniaturePainting} setting, the \emph{painter} -- or, equivalently, the \emph{player} -- is afforded greater freedom in choosing strokes or moves.
This proves very useful in developing the results of this work.

\smallskip
The paper is structured as follows.
The next section fixes basic concepts.
Section~\ref{sec:equivalence} establishes the equivalence between {\miniaturePainting} and {\freeFloodIt} and derives resulting complexity consequences.
In Section~\ref{sec:canonical_paint_plans}, we derive structural insights for {\cogem}-free graphs to be used in Section~\ref{sec:polynomial_time_algorithm} to develop a polynomial-time algorithm.
Section~\ref{sec:conclusion} concludes this paper.
To keep the presentation accessible, all proofs are deferred to Section~\ref{sec:proofs}.

%

\section{Preliminaries}
\label{sec:preliminaries}

Before presenting our results, we fix the notation used throughout the paper.
We mostly use standard math and graph notation.
For natural numbers $a,b$ with $a \leq b$ the set $\{a, a+1, \dots, b\}$ is denoted by $[a,b]$.
For sets $A$ and $B$, we denote the \emph{union}, \emph{difference}, and \emph{intersection} by $A+B$, $A-B$, and $A \cap B$, respectively.
Singleton sets $\{x\}$ are identified with the only contained element $x$, if this becomes clear from the context.

\subsection*{Graph Notation}

All considered graphs are finite, undirected, and without multiple edges or loops.
For every vertex $x \in V(G)$, we define the \emph{open neighborhood} $N(x) = \{y \mid xy \in E(G)\}$ and $N[x] = N(x)+x$ is the \emph{closed neighborhood}.
We say that $x$ \emph{sees} $y$, if $y \in N(x)$.
A \emph{dominating set} $D$ of a graph $G$ is a subset of $V(G)$ that is adjacent to all vertices, that is, $\bigcup_{x \in D} N[x] = V(G)$.

A \emph{path} $P_k$ on $k$ vertices is a graph on vertices $v_1, \dots, v_k$ and with exactly the edges $v_iv_{i+1}$ for all $i \in [1,k-1]$.
The vertices $v_1$ and $v_k$ are the \emph{endpoints} of the path.
If just the edge $v_kv_1$ is additionally present then we get the \emph{cycle} $C_k$ on $k$ vertices.

If $G$ is a graph then every subgraph $H=(V,E)$ of $G$ is defined by a vertex subset $V \subseteq V(G)$ and an edge subset $E \subseteq E(G)$ that consists only of edges with both endpoints in $V$.
A subgraph $H$ is \emph{induced}, if $E$ contains exactly all the edges $vw$ of $E(G)$ with $v,w \in V$.
For any vertex subset $V \subseteq V(G)$, we use $G[V]$ to denote the induced subgraph $H = (V,E)$ of $G$.
We use $G-V$ as a shorthand for $G[V(G)-S]$.
Graphs $G$ and $H$ are isomorphic, if a bijection $\sigma: V(G) \rightarrow V(H)$ exists with $xy \in E(G)$ if and only if $\sigma(x)\sigma(y) \in E(H)$.
If no induced subgraph of $G$ is isomorphic to a graph $H$ then $G$ is $H$-\emph{free}.
A subgraph $H$ of $G$, induced or not, is called \emph{dominating} $G$, if $V(H)$ is a dominating set of $G$.

As usual, an \emph{$x,y$-path} in $G$ is a subgraph of $G$ that is isomorphic to a path with endpoints $x$ and $y$.
The graph $G$ is \emph{connected}, if there is an $x,y$-path in $G$ for all distinct vertices $x,y \in V(G)$.
Otherwise, $G$ is called \emph{disconnected} and consists of \emph{connected components}, maximal induced subgraphs that are non-empty and connected.
A vertex subset $V \subseteq V(G)$ is (dis-) connected, if the subgraph $G[V]$ induced by $V$ is (dis-) connected.
Vertex sets $S \subseteq V(G)$ of connected graphs $G$ are called \emph{separator}, if $G-S$ is disconnected.

Though not central to this paper, we briefly mention some graph classes for completeness. 
The connected cycle-free graphs are the \emph{trees}. 
A graph is \emph{AT-free} if it contains no \emph{asteroidal triple}, that is, three vertices such that for every pair $x,y$ there is an $x,y$-path avoiding the neighborhood $N[z]$ of the third. 
\emph{Co-comparability graphs} are graphs whose complements admit a transitive orientation, that is, an orientation of the edges so that $u \to w$ whenever $u \to v$ and $v \to w$.
\emph{Split graphs} are graphs $G$ with a vertex set $C \subseteq V(G)$ such that $G[C]$ is a clique and $G - C$ is independent.
Finally, \emph{grid graphs} are graphs whose vertices can be placed on a two-dimensional integer lattice, with edges connecting precisely the horizontally or vertically adjacent vertices.

Like suggested above, cographs and {\cogem}-free graphs are the most important graph classes in here.
The following structural properties of graphs in these two classes are well-understood and often considered trivial observations.
However, we explicitly apply them in Sections~\ref{sec:canonical_paint_plans} and~\ref{sec:polynomial_time_algorithm} and, thus, include them  here, for the sake of being self-contained.
\begin{restatable}{observation}{ObsCographSeparatorPThree}
	\label{obs:cograph_separator_pthree}
	If $G$ is a connected cograph with separator $S$ then every pair of vertices in different components of $G-S$ have a common neighbor in $S$.
\end{restatable}
\begin{restatable}{observation}{ObsCographSeparatorDominate}
	\label{obs:cograph_separator_dominate}
	Every separator of a connected cograph $G$ is a dominating set of $G$.
\end{restatable}
\begin{restatable}{observation}{ObsCographSeparatorCover}
	\label{obs:cograph_separator_cover}
	If $G$ is a connected cograph with disjoint separators $S$ and $T$ then $S + T = V(G)$.
\end{restatable}
\begin{restatable}{observation}{ObsCogemfreeDominatingPFour}
	\label{obs:cogem-free_dominatingPfour}
	Every induced $P_4$ in a {\cogem}-free graph is a dominating set.
\end{restatable}

\subsection*{Graph Painting}

For obvious reasons, we call any (partial) function $p\colon V \rightarrow C$ with suitable color set $C$ a \emph{painting} of a graph $G=(V,E)$ instead of using the more common name \emph{coloring}.
To refer to a specific painting $p$ of a graph $G$ in short, we often simply say that $G$ is painted with $p$.
Notice that paintings may leave the color of vertices of $G$ undefined.
Accordingly, a painting $p$ is called \emph{total}, if the function $p$ is total (thus, defines a color for every vertex of $G$).
For every color $c \in C$, we use $p^{-1}(c) = \{v \in V(G) \mid p(v) = c\}$ to address all $c$-painted vertices of $G$.
Notice that the set $p^{-1}(c)$ is usually not connected.

For a painting $p$ of $G$, a subset $A \subseteq V$ is \emph{monochromatic}, if $p(v) = p(w)$ for all $v,w \in A$.
We point out that this notion includes undefined vertices: $A$ is also called monochromatic, if $p(v)$ is undefined for all $v \in A$.
In a painting $p$, a monochromatic set $A$ is called \emph{color component} of $G$, if $A$ is connected in $G$ and no superset of $A$ is monochromatic and connected.

\subsection*{Problem Statement}

Before we start with our contribution, we like to formally state the {\miniaturePainting} problem as set up in the introduction.
Let $G$ be a given graph and let $t\colon V(G) \rightarrow C$ be a template for $G$, that is, a total painting of $G$.
A \emph{brush stroke} (or just \emph{stroke}, for short) is a pair $(A, c)$ of a \emph{\area}, a vertex set $A\subseteq V(G)$ that is connected in $G$, and a color $c \in C$.
If $p\colon V(G) \rightarrow C$ is a painting of $G$ then the result of applying a stroke $(A, c)$ is a new painting $p'$ that, for all vertices $v \in A$ that are in the {\area} of the stroke, defines the new color $p'(v) = c$, and, for all unaffected vertices $v \in V(G)-A$, leaves the color $p'(v) = p(v)$ unchanged.

Based on this definition, a \emph{{\plan}} $P$ that paints $t$ onto $G$ in $s$ strokes is a non-empty sequence $P = (A_1, c_1), \dots, (A_s,c_s)$ of strokes that progressively creates a sequence of paintings $p_0, p_1, \dots, p_s$ such that (i) $p_0$ is the function that is undefined for all vertices, (ii) $p_i$ is the result of applying stroke $(A_i, c_i)$ on painting $p_{i-1}$, for all $i \in [1,s]$, and (iii) the painting resut $p_s$ is exactly the template $t$.

A {\plan} for $G$ is \emph{optimal}, if it uses the least amount of strokes $s$.
We also say that $G$ can be painted with $t$ in $s$ strokes, if an optimal {\plan} for $G$ has $s$ strokes.
The \textsc{Minimum} {\miniaturePainting} problem is the computational task to find an optimal {\plan} for a given graph $G$ and a template $t$.
A decision version of this is $k$-{\miniaturePainting} that just asks, for given $(G,t)$, if $G$ can be painted with $t$ in $k$ strokes or less.
In the remainder of the paper, it is always obvious from the context when we consider the optimization or decision version of the problem.
Thus, it is convenient for us to just use the name {\miniaturePainting}, in every situation.

\medskip
To better organize our subsequent work with {\plans} $P = (A_1,c_1), \dots, (A_s,c_s)$, we also introduce the following notions.
For all $v \in V(G)$, we use the \emph{finalizing} $f_P(v) = i$ to refer to the maximum index $i$ of a stroke that overdraws $v$, that is, with $v \in A_i$ and $v \not\in A_j$ for all $j > i$.
Obviously, this implies that $t(v) = c_i$.
Moreover, we use $F_P(i)$ for the set of vertices that are finalized by stroke $i$, that is $F_P(i) = \{v \in V(G) \mid f_P(v) = i\}$.

%

\section{Wait, how is that the {\freeFloodIt} Game?}
\label{sec:equivalence}
	
At first glance, {\freeFloodIt}, as first described by Clifford~et~al.~\cite{Clifford2010,Clifford2012}, seems to have little in common with {\miniaturePainting}.
Ultimately, the rules of this single player game demand essentially the opposite of our setting:
Given a graph $G = (V, E)$ that is already painted via $p_0\colon V \rightarrow C$, the player has to define a (possibly empty) \emph{flooding sequence} $(v_1,c_1), \dots, (v_s,c_s)$ of moves $(v_i, c_i)$ that specify a \emph{pivot vertex} $v_i \in V$ and a \emph{flooding color} $c_i$.
The result of a flood move $(v_i, c_i)$ is a new painting $p_i\colon V \rightarrow C$ of $G$ that alters $p_{i-1}$ as
\[p_i(v) = \begin{cases}
	c_i, & \text{if $v_i$ and $v$ are in the same color component with respect to } p_{i-1},\\
	p_{i-1}(v), &\text{otherwise.}
\end{cases}\]
The goal of the game is to make $p_s$ monochromatic, that is, a constant function, while minimizing the sequence length $s$.

It is precisely the \emph{oppositeness} of the goals -- painting versus decoloring a given graph -- that constitutes the commonality between the two optimization problems.
This section shows that reversing any flooding sequence for a given graph $G$ painted with $p_0$ also provides a suitable painting strategy for $G$ when $p_0$ is given as the template.
Whenever a move $(v_i, c_i)$ creates a bigger {\area} $A$ in $G$ with color $c_i$ by $c_i$-flooding a smaller $c$-colored {\area} $B$ that contains $v_i$ then, for {\miniaturePainting}, we can reverse this move with the paint stroke $(B, c)$.
Something similar applies in the other direction, too, but it is a little more complicated.
Here, we will have to answer how a paint stroke can be reversed into a single flooding move, if the stroke covers multiple differently colored vertices.
To address this question, the following begins with a normalization for painting strategies that circumvents the possibility of falling into said trap.

More precisely, we introduce the following normalized version of {\miniaturePainting}:
In a painting $p$ of a graph $G$, a stroke $\sigma = (A, c)$ is called \emph{\recursive} when (i) the stroke $\sigma$ paints a monochromatic {\area} $A$ in $p$, that is, $p(v) = p(w)$ for all $v,w \in A$, and (ii) in the painting $p'$ derived from $p$ by applying stroke $\sigma$, the {\area} $A$ is a color component of $G$.
A {\plan} $(A_1,c_1), \dots, (A_s,c_s)$ for a graph $G$ that, step by step, creates the paintings $p_0, p_1, \dots, p_s$ is called \emph{\recursive}, if it satisfies the following two properties:
Firstly, all strokes $(A_i, c_i), i \in [1,s]$ have to be \emph{\recursive}, that is, the {\area} $A_i$ is monochromatic in $p_{i-1}$ and a color component of the graph $G$ painted with $p_i$.
Secondly, the first stroke has to fill the entire graph, that is, $A_1 = V(G)$.
Observe that the {\plan} for the example in Figure~\ref{fig:example}~(a) that the introduction presents is a {\recursive} one.

\begin{restatable}{lemma}{LmaRecursivePainting}
	\label{lma:recursive_painting}
	A template $t$ can be painted on a graph $G$ in $s$ strokes, if and only if $G$ has a {\recursive} $s$-stroke {\plan} for $t$.
\end{restatable}

With the above normalization tool in hand, we are ready to provide the quasi-equivalence of {\miniaturePainting} and {\freeFloodIt}.
\begin{restatable}{theorem}{ThmPaintingFloodingEquivalence}
	\label{thm:painting_flooding_equivalence}
	Let $G = (V,E)$ be a graph and $t\colon V \rightarrow C$ be a template for $G$.
	There is an $s$-{\plan} that paints $t$ onto $G$ if and only if there is a flooding sequence of $s-1$ moves for $G$ when initially painted with $t$.
\end{restatable}

As a reading of Section~\ref{sec:proofs} reveals, the sequence from the introduction that floods the example in Figure~\ref{fig:example}~(a) is, in fact, obtained by applying the \emph{reversion} approach from the proof of Theorem~\ref{thm:painting_flooding_equivalence} to the previously given {\plan}.
And like the introduction suggests, it is true that the previous result ties {\miniaturePainting} and {\freeFloodIt} tightly together.
It actually allows transferring many known results from the latter problem to the former.
Hardness of {\freeFloodIt} translates directly to {\miniaturePainting} and, so, we can formally derive the following (partial) list of conclusions:
\begin{corollary}[see~\cite{Clifford2012,FELLOWS2015,Fukui2013}]\label{cor:painting_hardness}
	\emph{\miniaturePainting} is $\NP$-complete, even when input graphs are restricted to
	\begin{itemize}
		\setlength{\itemsep}{0pt}
		\setlength{\parskip}{1pt}
		\item grid graphs, even with templates of just three colors,
		\item trees, even with templates of just three colors,
		\item AT-free graphs,
		\item co-comparability graphs, and
		\item split graphs.
	\end{itemize}
\end{corollary}
As indicated in the introduction, the list of tractable cases is much shorter and we will address it in the next section.

%

\section{Co-gem-free Graphs allow Canonical Paint Plans}
\label{sec:canonical_paint_plans}

To the best of our knowledge, {\freeFloodIt} is known to be solvable in polynomial time on cographs, yet an explicit proof appears to be missing from the literature. 
This is likely because a dynamic programming approach using the cotree of a cograph is straightforward to envision. 
However, there exists an even simpler strategy, which we later generalize to {\cogem}-free graphs. 
It is therefore worthwhile to present this approach in detail.

The first key observation is that a disconnected cographs $G$ can be flooded by processing their connected components $G_1, \dots, G_k$ sequentially.
Each of these components is again a cograph, and moreover connected.
For connected cographs, a second crucial insight comes into play: they contain many \emph{dominating edges}, that is, edges $vw$ such that $V(G) = N[v] + N[w]$.
An optimal flooding sequence can then be obtained by exhaustively considering such dominating edges $vw$ and choosing one that, after flooding $vw$ monochromatically, yields the shortest possible flooding sequence for its neighborhood -- which in this case coincides with the entire graph.

We put this basic idea into a formal statement but with respect to mirror problem of {\miniaturePainting} on cographs.
\begin{restatable}{theorem}{ThmPaintingCographs}
	\label{cor:painting_cographs}
	There is a polynomial time algorithm for \emph{\miniaturePainting} on cographs.
\end{restatable}
More important than the result itself is the structural property of cographs that enables this efficient approach.
While {\cogem}-free graphs no longer necessarily admit dominating edges -- as the path $P_5$ illustrates convincingly -- Observation~\ref{obs:cogem-free_dominatingPfour} provides an adequate structural substitute.
Instead of relying on dominating edges, we exploit induced $P_4$ that serve as dominating sets.

Before diving into the details of our approach, we briefly motivate how a dominating $P_4$ can be leveraged to paint a graph~$G$ that contains it.
In fact, any small vertex set $D \subseteq V(G)$ that is both, dominating and connected, can be applied in a promising two-phase {\plan}:
In the first phase, the set~$D$ is used as a hub to paint, for every color $c \in C$, the {\area} $A_c = t^{-1}(c) + D$, that is, the union of~$D$ and all vertices that are assigned color~$c$ by the template~$t$.
Since $D$ is dominating, we have $A_c \subseteq \bigcup_{v \in D} N[v]$, and the connectedness of~$D$ ensures that $A_c$ is indeed a {\area}.
After the first phase, all of $G$ but $D$ has already been finished.
The second phase then completes the painting of the vertices in~$D$ itself, which can be achieved using at most~$|D|$ additional strokes.

One advantage of this strategy is that the order in which colors are applied during the first phase is largely arbitrary.
Consequently, the main optimization task is reduced to determining an efficient way of painting the vertices of~$D$.
Another benefit is that the entire plan requires no more than $|C| + |D|$ strokes, thereby providing a useful upper bound for the search for even better strategies.

This idea is particularly appealing for {\cogem}-free graphs, as they tend to admit many small connected dominating sets.
Assuming that $G$ is connected, it is either a cograph, in which case it contains dominating edges, or it contains induced $P_4$, all of which are dominating.
Hence, any optimal painting strategy for~$G$ uses at most $|C| + 4$ strokes.
In fact, we believe that every connected {\cogem}-free graph admits an optimal painting strategy that strictly follows the two-phase scheme of $|C|$ plus less than five strokes described above.
In this paper, however, we establish this claim only in a slightly weakened form.

We begin by formalizing the generalized two-phase concept.
Let $G$ be a graph with a connected dominating set $D \subseteq V(G)$, let $k$ be a non-negative integer, and let $\prec$ be any total ordering of the color set~$C$.
An $s$-stroke {\plan} $P = (A_1,c_1), \dots, (A_s,c_s)$ is called \emph{$(D,k,\prec)$-canonical} if the following conditions hold:
\begin{enumerate}[label=(\roman*)]
	\item for every $i \in [1,s-k]$, the stroke $(A_i,c_i)$ paints the {\area} $A_i = t^{-1}(c_i) + D$,
	\item for all $i,j \in [1,s-k]$ with $i < j$, we have $c_i \prec c_j$.
\end{enumerate}
Thus, a $(D,k,\prec)$-canonical plan consists of a first phase of $s-k$ strokes that explicitly exploits the hub~$D$ to finish all vertices in $V(G) - D$ according to the prescribed color order, followed by a second phase of $k$ strokes that completes the painting of~$G$ without further use of~$D$ as a hub.

We emphasize that the inclusion of the ordering~$\prec$ is not a methodological restriction.
Rather, it will later provide the algorithmic freedom to fix an arbitrary color order in advance.
Also note that it is not required that $k = |D|$.
The first $s-k$ strokes of a $(D,k,\prec)$-canonical plan~$P$ are called the \emph{head} of~$P$, while the remaining $k$ strokes form its \emph{tail}.

\medskip
The following lemma captures the relationship between {\cogem}-free graphs and canonical {\plans}.
\begin{restatable}{lemma}{LemCanonicalOptplan}
	\label{lem:canonical_optplan}
	For every {\cogem}-free graph $G$ containing an induced $P_4$, every template $t\colon V(G) \rightarrow C$, and every total ordering $\prec$ of~$C$, there exists an optimal {\plan} that is $(D,k,\prec)$-canonical for some dominating set $D$ that induces a $P_4$ in~$G$ and some $k \leq 12$.
\end{restatable}
The remainder of this section prepares the proof of Lemma~\ref{lem:canonical_optplan}.
At a high level, the proof proceeds in two steps.
First, we consider the set of all optimal {\plans} for~$G$ and impose a carefully chosen ordering on them.
Second, we show that a maximal element with respect to this ordering is, up to minor adjustments, canonical.

To motivate the ordering, observe that a $(D,k,\prec)$-canonical plan in the lemma uses an induced $P_4$, namely~$D$, as a hub to finish all other vertices during the first phase.
Consequently, all vertices of~$D$ are finalized during the second, $k$-stroke phase.
Thus, in such a plan, there exists an induced $P_4$ whose vertices are painted very late.
The idea behind the ordering is to prefer plans that postpone the completion of an induced $P_4$ as far as possible.

We formalize this intuition by assigning a numerical index to every optimal {\plan} $P$.
Let $s$ be the number of strokes in every optimal plan and, likewise, in $P$.
For every induced $P_4$~$\pi$ of~$G$, we define a corresponding index as follows.
Let $a,b,c,d$ denote the vertices of~$\pi$, ordered by their finalizing indices, that is, $f_P(a) \leq f_P(b) \leq f_P(c) \leq f_P(d)$.
We then define the \emph{finishing index} of~$\pi$ in~$P$ as
\[f_P(\pi) = s^3 \cdot f_P(a) + s^2 \cdot f_P(b) + s \cdot f_P(c) + f_P(d).\]
This value is primarily determined by the finalization of~$a$, secondarily by that of~$b$, then~$c$, and finally~$d$.
The \emph{finishing index} of the plan~$P$ is defined as the maximum finishing index among all induced $P_4$ in~$G$, that is,
\[f_P = \max\{ f_P(\pi) \mid \pi \text{ is an induced } P_4 \text{ in } G \}.\]
In the proof of Lemma~\ref{lem:canonical_optplan}, this index is used to order the optimal {\plans}. 
We focus on a plan that maximizes $f_P$, intuitively delaying the finishing of some induced $P_4$ as much as possible. 
The central argument shows that such a plan can be transformed into a $(D,k,\prec)$-canonical plan for any induced $P_4$ $D$ with maximum finishing index $f_P(D) = f_P$. 
After minor adjustments, the approach demonstrates that between consecutive finishing strokes of vertices from $D$, there can be at most two strokes that do not finish any vertex of $D$. 
Equivalently, the stroke that completes $D$ can be followed by at most two additional strokes. 
This establishes the maximum distance of twelve between the first stroke that finishes a vertex of $D$ and the last stroke of the plan. 
Any violation of these limits would imply the existence of another optimal plan with a finishing index below $f_P$, yielding a contradiction.

The proof heavily relies on the flexibility of {\plans}, which is not available in the reversed strategies of the {\freeFloodIt} game. 
In fact, a frequently used operation is the liberal adjustment of {\areas} in certain strokes to enforce a desired property of a plan. 
By contrast, flooding moves only allow pivot selection, which rigidly fixes the flooded area. 
Thus, while analogous canonical sequences must exist for {\freeFloodIt}, working with them would be considerably more involved.

The following section puts $(D,k,\prec)$-canonical {\plans} to work, developing a polynomial-time algorithm that exploits their simple structure.

%

\section{Brute-forcing Canonical Paint Plans works in Polynomial Time}
\label{sec:polynomial_time_algorithm}

As the title of this section suggests, the algorithm presented here does not rely on any sophisticated procedure to construct an optimal {\plan}.
Instead, it generates a polynomial number of structured candidate plans and verifies them exhaustively.
More precisely, we enumerate all $(D,12,\prec)$-canonical {\plans} for all suitable choices of $D$, $k$, and~$\prec$.
In fact, by Lemma~\ref{lem:canonical_optplan}, we may assume that $|D|=4$ and $k\leq 12$.
Accordingly, we restrict our attention to $(D,k,\prec)$-canonical {\plans} with $|D|=4$ and $k=12$.
In the following, we refer to these plans simply as \emph{$(12,\prec)$-canonical {\plans}}.

To generate such plans, we first construct a polynomial-size set of possible \emph{heads} and a polynomial-size set of possible \emph{tails}.
By combining each head with each tail, we obtain a collection of candidate {\plans}.
If any candidate passes a verification step, it yields an optimal solution.

\paragraph{Heads.}
We begin with the construction of possible heads.
Fixing an (arbitrary) total ordering~$\prec$ of the colors already eliminates many essentially equivalent strategies.
Indeed, given a $(12,\prec)$-canonical {\plan}~$P$, its head~$H$ is uniquely determined by its dominating set~$D$ and the set~$F\subseteq C$ of colors used in~$H$.
Once $D$ and~$F$ are fixed, the head can be reconstructed by enumerating the colors $c\in F$ in increasing order with respect to~$\prec$ and applying the strokes $(t^{-1}(c)+ D,\,c)$.

Since any graph~$G$ contains at most $\binom{|V(G)|}{4}\in O(|V(G)|^4)$ induced $P_4$s, we can simply enumerate all $4$-vertex subsets and test whether they induce a $P_4$.
Likewise, if a $(D,12,\prec)$-canonical {\plan} has length~$s$, then its head uses exactly $s-12$ colors.
Because $s\leq |C|+4$, we only need to consider $\binom{|C|}{|C|+4-12} \in O(|C|^8)$ choices for~$F$.

These observations yield the following result.
\begin{restatable}{lemma}{LemHeads}
	\label{lem:heads}
	Let $G$ be a graph with template~$t$, and let $\prec$ be any ordering of the colors used in~$t$.
	The set
	\[
	H_s := \{ \operatorname{head}(P) \mid \text{$P$ is a $(12,\prec)$-canonical {\plan} of length $s$} \}
	\]
	can be generated in polynomial time.
\end{restatable}

\paragraph{Tails.}
We now turn to the construction of possible tails.
Unlike the head, the tail of a $(12,\prec)$-canonical {\plan} is not structurally rigid.
Naively enumerating all possibilities would therefore lead to a prohibitive blow-up.

To control this, we impose an additional maximality condition.
Let $P=(A_1,c_1),\dots,(A_s,c_s)$ be a $(D,k,\prec)$-canonical {\plan}.
We call $P$ \emph{maximal} if there is no index $i \in [1,s]$ and vertex $v\in V(G)$ such that $A_i\subset A_i+ v$ and replacing $A_i$ by $A_i+ v$ yields another valid $(D,k,\prec)$-canonical {\plan}.
Intuitively, in a maximal plan, every stroke absorbs all vertices that can be added without violating canonicity.

Crucially, this notion of maximality only affects the tail.
Expanding any stroke in the head would contradict the definition of a canonical {\plan}.
Hence, Lemma~\ref{lem:heads} remains valid even when we restrict attention to maximal plans.

\begin{restatable}{observation}{ObsMaximal}
	\label{obs:maximal}
	For every $(D,k,\prec)$-canonical {\plan} of length~$s$, there exists a maximal $(D,k,\prec)$-canonical {\plan} of the same length.
\end{restatable}

We therefore restrict our attention to maximal $(12,\prec)$-canonical {\plans}.
Since the tail consists of exactly twelve strokes, we can brute-force suitable ``seed'' vertices and colors for these strokes and then expand each area maximally.
This yields a polynomial-size superset of all possible tails.

\begin{restatable}{lemma}{LemTails}
	\label{lem:tails}
	Let $G$ be a graph with template~$t$, and let $\prec$ be any ordering of the colors used in~$t$.
	A set
	\[
	T_s \supseteq \{ \operatorname{tail}(P) \mid \text{$P$ is a maximal $(12,\prec)$-canonical {\plan} of length $s$} \}
	\]
	can be generated in polynomial time.
\end{restatable}

Notice that $T_s$ may contain tails that do not belong to any canonical {\plan}.
But this is harmless: such candidates will simply fail to combine with any head.
Most importantly, $T_s$ contains only polynomially many elements, each of constant size.

Combining heads and tails yields the following.
\begin{restatable}{lemma}{LemHeadsAndTails}
	\label{lem:heads_and_tails}
	Let $G$ be a graph with template~$t$, and let $\prec$ be any ordering of the colors used in~$t$.
	The set
	\[
	P_s := \{ P \mid \text{$P$ is a maximal $(12,\prec)$-canonical {\plan} of length $s$} \}
	\]
	can be generated in polynomial time.
\end{restatable}

We are now ready to present the polynomial-time algorithm for {\miniaturePainting} on {\cogem}-free graphs.

\begin{algorithm}
	\caption{Polynomial-time algorithm for {\miniaturePainting} on {\cogem}-free graphs}
	\label{alg:co-gem}
	\begin{algorithmic}[1]
		\Procedure{solve}{$G,t$}
			\If{$G$ is a cograph}
				\State \Return \Call{solve\_cograph}{$G,t$}
			\EndIf
			\For{$i = c$ to $c+3$}
				\State $S \gets$ \Call{generate\_P$_s$}{$G,t,i$}
				\If{$S \neq \emptyset$}
					\State \Return an arbitrary element of $S$
				\EndIf
			\EndFor
		\EndProcedure
	\end{algorithmic}
\end{algorithm}

We conclude this section with our main result.
\begin{restatable}{theorem}{ThmMainResult}
	\label{thm:main_result}
	Let $G$ be a {\cogem}-free graph and let $t$ be a color template.
	An optimal {\plan} can be found in polynomial time.
\end{restatable}

%

\section{Conclusion}
\label{sec:conclusion}

We have introduced the {\miniaturePainting} problem as a tool to analyze and solve instances of {\freeFloodIt} on structured graph classes. 
By formalizing canonical painting plans and exploiting the structure of small connected dominating sets, we developed a polynomial-time algorithm for {\cogem}-free graphs. 
Along the way, we also recover a simple polynomial-time solution for cographs.

Looking forward, our results suggest several directions for further research. 
One is to reduce the exponent of the current polynomial-time algorithm.
Another is to explore whether the concept of connected dominating sets and canonical painting plans can be adapted to other graph classes, potentially identifying new tractable cases of {\freeFloodIt}.

%

\section{Technical Proofs}
\label{sec:proofs}

This section collects the proofs of all theorems, lemmas, and observations from the main text, presented in the same order as they appear.
Each result is restated for clarity before its proof.

\ObsCographSeparatorPThree*
\begin{proof}
	Let $v$ and $w$ be in different connected components of $G-S$ and consider any $v,w$-path $P$ in $G$ with the fewest number of vertices.
	Since $P$ must contain at least one vertex of $S$ and because $G$ is $P_4$-free, the reader easily verifies that $P$ is an induced $P_3$ on vertices $v$, $w$ and one more vertex from $S$.
\end{proof}

\ObsCographSeparatorDominate*
\begin{proof}
	For any pair of vertices $v$ and $w$ in different connected components of $G-S$, Observation~\ref{obs:cograph_separator_pthree} implies a common neighbor of $v$ and $w$ in $S$.
	Hence, $S$ specifically dominates $v$ and $w$ and, by the arbitrarity of choice, $S$ is generally a dominating set of $G$.
\end{proof}

\ObsCographSeparatorCover*
\begin{proof}
	If there was $x \in V(G) - (S + T)$ then we could choose a vertex $v$ in the graph $G-S$ that is not in the same connected component as $x$ and, similarly, a vertex $z$ from another component of $G-T$.
	According to Observation~\ref{obs:cograph_separator_pthree}, there are vertices $w \in S$ and $y \in T$ such that $v,w,x$ and $x,y,z$ are induced $P_3$ in $G$.
	The vertices $v,w,x,y,z$ are distinct by preconditions, choice, and because $v = z$ implies an $x,z$-path $x,w,v=z$ that does not pass $T$.
	If the edges $wz$ or $vz$ existed there would be $x,z$-paths $x,w,z$ or $x,w,v,z$ that do not pass $T$ and an edge $vy$ would imply the $v,x$-path $v,y,x$ that misses $S$,
	Now, independent of $wy \in E(G)$, there is an induced $P_4$ in $G$, namely $v,w,y,z$ or $v,w,x,y$.
	This is a contradiction to the assumption that $x$ exists.
\end{proof}

\ObsCogemfreeDominatingPFour*
\begin{proof}
	If there was a non-dominating $P_4$ in a {\cogem}-free graph $G$, say on the vertices $a,b,c,d$ then there would be a vertex $x$ that is not adjacent to any of $a,b,c,d$ and, hence, $\{a,b,c,d,x\}$ would induce a {\cogem}, a contradiction.
\end{proof}

\LmaRecursivePainting*
\begin{proof}
	We only have to show that an $s$-stroke {\plan} for $t$ implies the existence of a {\recursive} one that has the same number of moves and still paints $t$ onto $G$.
	For a start, we argue the existence of a plan that solely consists of {\recursive} strokes.
	So, let $P = (A_1,c_1), \dots, (A_s,c_s)$ be any $s$-stroke {\plan} for $t$ that minimizes $j$, the last (that means, the maximum) index of a non-{\recursive} stroke.
	We would simply let $j = 0$, in case $P$ is free of non-{\recursive} strokes.
	
	For a contradiction, assume that, $j > 0$ and let $(A_j, c_j)$ be the respective stroke in $P$.
	Moreover, we define $A$ as that superset of {\area} $A_j$ that is a $c_j$-colored component of $G$ in the painting $p_j$.
	We create a new {\plan} $P'$ from $P$ by sequencing, for all $i \in [1,s]$, the strokes $(A'_i,c_i)$, where $A'_i = A_i+A$, if $i \leq j$ and $A_i \cap A$ is not empty, and, otherwise, $A'_i = A_i$.
	Apparently, $A'_1, \dots, A'_s$ are valid {\areas} because $A_1, \dots, A_s$ and $A$ are connected vertex sets.
	In fact, this makes every $A'_i$ a connected vertex set, in particular when $A'_i = A_i + A$, since then $A_i \cap A$ is not empty.
	Let $p'_0, p'_1, \dots, p'_s$ be the paintings of $P'$.
	
	We conclude that the paintings $p_j$ and $p'_j$ are the same as follows:
	Clearly, $A = A'_j$ and, so, for all $v \in A$, we have $p_j(v) = c_j$ by the definition of $A$ and $p'_j(v) = c_j$ by the stroke $(A'_j, c_j)$ in $P'$.
	For every vertex $v \in V(G)-A$, consider the last stroke $(A_i, c_i)$ in $P$ (that means with maximum $i$) where $v \in A_i$ and $i \leq j$.
	Notice that $i$ is not necessarily $f_P(v)$.
	Since $A_i - A = A'_i - A$, we get that $P'$ also paints $v$ in stroke $i$.
	Now consider the last stroke $(A'_{i'}, c_{i'})$ in $P'$ (that means with maximum $i'$) where $v \in A'_{i'}$ and $i' \leq j$.
	Again, $i'$ is usually not $f_{P'}(v)$.
	By the same argument as before, we find that $P$ also paints $v$ in stroke $i'$.
	Since both, $i$ and $i'$, are maximal indices within $[1,j]$ where $v$ is painted, it has to be $i=i'$.
	Thus, $p_j(v) = p'_j(v) = c_i$ and we have the equivalence of $p_j$ and $p_j'$.

	By definition, $(A_i, c_i) = (A'_i, c_i)$ for all $i \in [j+1,s]$ and, thus, $p'_i$ is the same as $p_i$ for every $i \geq j$.
	Hence, $P'$ still paints the template $t = p_s = p'_s$ and, for all $i > j$, the stroke $i$ is {\recursive} in $P'$ just like in $P$.
	
	We derive the contradiction by arguing that stroke $(A'_j,c_j)$ is {\recursive} in $P'$.
	But this is fairly obvious.
	Firstly, every previous stroke in $P'$ either evades or fills $A'_j = A$ entirely.
	This means that, with respect to $p'_{j-1}$, the {\area} $A'_j$ is monochromatic.
	Moreover, $A$ has been chosen as the superset of $A_j$ that is a $c_j$-colored component of $G$ with respect to painting $p_j$.
	Since $p'_j = p_j$, we can conclude that $A'_j$ is a color component of $p'_j$, too.

	Consequently, $P'$ is an $s$-stroke {\plan} for $t$ where the maximum index of a non-{\recursive} stroke is lower than $j$.
	This is a contradiction to the selection of $P$ and, ultimately, to the assumption of $j > 0$.
	Hence, $P$ must have been made of {\recursive} strokes, only, in the first place.

	\medskip
	It remains to turn $P$ into a {\recursive} {\plan}.
	For this, we simply replace the first stroke $(A_1, c_1)$ with $(V(G), c_1)$, which is still a {\recursive} one, as the reader easily verifies.
	The effect of this modification is rather a global one as it has the potential of changing all paintings $p_1, \dots, p_s$.
	But it is easy to see that the impact is also mild and cannot change the semantic.
	In fact, the only change happens, if a stroke $(A_i,c_i)$ painted an {\area} that has been monochromatically undefined in $p_{i-1}$ before the modification.
	Afterwards, the stroke simply paints the same {\area} $A_i$, but it is monochromatically $c_1$-colored.
	The result with respect to $A_i$ is the same -- in $p_i$ it is a $c_i$-colored component of $G$.
	Since $t$ is a total function, the modification preserves $t = p_s$.
\end{proof}

\ThmPaintingFloodingEquivalence*
\begin{proof}
	With the help of Lemma~\ref{lma:recursive_painting}, we are in the stronger position where we can even work with {\recursive} {\plans}.
	We argue via induction on $s$ and, for the start, we let $s = 1$.
	Then, there is a {\plan} for $t$ that consists of a single stroke $(A_1, c_1)$, if and only if it is already a {\recursive} plan, if and only if $t$ is constant function that assigns color $c_1 = t(v)$ to all vertices $v \in V = A_1$, if and only if $G$ is already monochromatic under the initial painting $t$, if and only if there is a flooding sequence of $s-1 = 0$ moves.

	\medskip
	For the induction step, let $s > 1$ and assume that the theorem is true for all $s'$-{\plans} and flooding sequences of $s'-1$ moves, respectively, where $s' < s$.
	By Lemma~\ref{lma:recursive_painting}, there is $s$-stroke {\plan} for $t$, if and only if there is a {\recursive} plan $P = (A_1,c_1), \dots, (A_s,c_s)$ that paints $t$ onto $G$.
	According to the definition of {\miniaturePainting}, the latter is true, if and only if there are paintings $p_0, p_1, \dots, p_s$ with the undefined function $p_0$ and $p_s = t$ and, for all $i \in [1,s]$, the painting $p_i$ is derived from $p_{i-1}$ by defining for all $v \in V$ that $p_i(v) = c_i$, if $v \in A_i$, and, otherwise, $p_i(v) = p_{i-1}(v)$.
	Apparently, this is equivalent to (i) the existence of the {\recursive}, $(s-1)$-{\plan} $P' = (A_1,c_1), \dots, (A_{s-1},c_{s-1})$ for the painting of $t' = p_{s-1}$ onto $G$ (notice that $t'$ is a total function and, thus, a valid template) and (ii) that $t$ is derived from $t'$ by letting $t(v) = c_s$ for all $v \in A_s$ and $t(v) = t'(v)$ for all $v \in V-A_s$.
	Since $P$ is {\recursive}, the {\area} $A_s$ is monochromatic in painting $t'$, and we refer to its color by $c = t'(A_s)$.
	Recall also that $A_s$ becomes a $c_s$-colored component of $G$ in the painting $t$.

	By the induction assumption and the definition of {\freeFloodIt}, the previous statement holds, if and only if (i) there is flooding sequence $S'$ of $s-2$ moves for $G$ when initially painted with $t'$ and (ii) making a flood move $(v,c)$ with $v \in A_s$ on the $t$-painted graph $G$ creates the painting $t'$ on $G$.
	Observe that part (ii) of the statement is valid, because $A_s$ is a $c_s$-colored component of $G$ and, hence, flooding this region with $c$, restores the monochromatic region $A_s$ of $t'$.
	Finally, by the definition of flooding sequences, the last claim is true, if and only if $S = (v,c) + S'$ is a $(s-1)$-move flooding sequence for $G$ when it is initially painted with $t$.
	Here, $(v,c)+S'$ means the sequence obtained from prepending the move $(v,c)$ to the sequence $S'$.
	This completes the proof.
\end{proof}

\ThmPaintingCographs*
\begin{proof}
	Let $G$ be a cograph and $t$ a template for $G$.
	We can assume that $G$ is connected as, otherwise, we could process the connected components individually.
	Either $G$ is a single vertex, which is a trivial case, or it contains dominating edges.
	This is because it is commonly known that a connected cograph $G$ is the full join of two cographs $G_1$ and $G_2$ (see for instance~\cite{Corneil1985}) and, clearly, every edge between $G_1$ and $G_2$ is dominating.

	If there is a dominating edge $v_1v_2$ where, without loss of generality, $t^{-1}(t(v_1))$ is connected, then the reader easily verifies that painting {\area} $t^{-1}(c) + \{v_1,v_2\}$ with $c$ for all $c \in C-t(v_1)$ in arbitrary order and, at the end, {\area} $t^{-1}(t(v_1))$ with $t(v_1)$ yields a $|C|$-stroke {\plan} for $G$.
	Because this is the minimum possible length of any {\plan} for $t$, it follows that it must be optimal.
	Otherwise, if such an edge does not exist, taking an arbitrary dominating edge $v_1v_2$ and appending the $|C|-1$ initial strokes of the above plan with the strokes $(t^{-1}(t(v_1))+v_2, t(v_1))$ and $(\{v_2\}, t(v_2))$ creates a $|C|+1$-stroke plan.

	Observe that the second case is optimal, as well.
	To see this, we only have to show that there is no plan with $|C|$ strokes.
	Such a $|C|$-stroke plan has only one stroke for every used color.
	So, the last stroke would necessarily have to paint the {\area} $t^{-1}(c)$ for some $c \in C$ because this has not happened before and it cannot override anything else.
	But now, taking any vertex $v_1 \in t^{-1}(c)$ and any vertex $v_2$ from the cograph of $G_1$ or $G_2$ that does not contain $v_1$ obviously presents a dominating edge of the above type.
	This is a contradiction.

	Obviously, all this works in polynomial time since we only have to check the $|E(G)|$ edges for (i) being dominating and (ii) the connectedness-property, which is done for every eadfe using linear-time BFS.
\end{proof}

\LemCanonicalOptplan*
\begin{proof}
	Let $P = (A_1, c_1), (A_2, c_2), \dots, (A_s,c_s)$ be an optimal {\plan} for $t$ on $G$ whose finishing index $f_P$ is maximal and let $D$ be an induced $P_4$ of $G$ whose finishing index with respect to $P$ is maximal.
	Note that neither $P$ nor $D$ need be unique and that $f_P = f_P(D)$.
	Let $a,b,c,d$ be the vertices of $D$, ordered by their finalizing indices, and set
	$\alpha = f_P(a)$, $\beta = f_P(b)$, $\gamma = f_P(c)$, and $\delta = f_P(d)$.
	Thus, $\alpha \leq \beta \leq \gamma \leq \delta$.

	We first observe that, without loss of generality, we may assume that for every stroke index $i \leq \alpha$ it holds that
	$A_i = t^{-1}(c_i) + D$.
	Indeed, suppose this is not the case for some stroke $(A_i,c_i)$.
	Replacing $A_i$ by $A'_i = t^{-1}(c_i) + D$ does not invalidate $P$ as a valid {\plan} for $t$.
	To see this, consider any vertex $v \in V(G)$ whose membership differs between $A_i$ and $A'_i$.
	If $f_P(v) = i$, then $v \in A_i$ and, hence, $t(v) = c_i$, which implies $v \in t^{-1}(c_i) \subseteq A'_i$, a contradiction.
	If $f_P(v) < i$, then repainting $v$ with color $c_i$ at stroke $i$ is correct and no subsequent stroke affects $v$.
	Finally, if $f_P(v) > i$ -- as is the case for all $v \in \{b,c,d\}$ -- then any change to the color of $v$ at stroke $i$ is irrelevant, since it will be overwritten by a later stroke.
	Thus, the modification preserves correctness.

	Based on this observation, we may further assume, again without loss of generality, that
	$i < j$ implies $c_i \prec c_j$ for all $i,j \in [1,s-\alpha-1]$.
	If this ordering is violated by two strokes $(A_i,c_i)$ and $(A_j,c_j)$ with $i<j$, then swapping these strokes yields a valid plan.
	Indeed, if the swap altered the coloring of any vertex after the first $\alpha-1$ strokes, then that vertex would necessarily lie in $D$, since all {\areas} up to index $\alpha$ intersect only in $D$.
	However, stroke $\alpha$ repaints all vertices of $D$ entirely, because $A_\alpha = t^{-1}(c_\alpha) + D$, and thus restores the correct colors.

	Consequently, to show that $P$ is $(D,k,\prec)$-canonical for some $k \leq 12$, it remains only to prove that $s - \alpha \leq 12$.
	The key idea is to show that between any two consecutive values among $\alpha$, $\beta$, $\gamma$, and $\delta$ there can be at most two intermediate strokes.
	To this end, let $p_0, p_1, \dots, p_s$ be the paintings created by $P$, thus, with $p_s = t$.

	We begin with the case of $\alpha$, which, on the one hand, is slightly simpler than the others and, on the other hand, cannot build on a property of the unfinished rest graph of $G$ that is crucial later.
	\begin{claim}\label{cla:canonical_optplan:one}
		If $\alpha < s$, then $\alpha + 1 \in \{\beta, \gamma, \delta, s\}$.
	\end{claim}
	Suppose, for the sake of contradiction, that $\alpha + 1 \notin \{\beta, \gamma, \delta, s\}$.
	Then, for every vertex $v \in D$, it clearly holds that either $f_P(v) = \alpha$ or $f_P(v) > \alpha + 1$.

	We now construct a new {\plan} $P'$ from $P$ by replacing strokes $\alpha$ and $\alpha+1$ with strokes $(A'_{\alpha}, c'_{\alpha})$ and $(A'_{\alpha+1}, c'_{\alpha+1})$, defined as $A'_{\alpha} = A_{\alpha+1} + D$, $c'_{\alpha} = c_{\alpha+1}$, $A'_{\alpha+1} = A_{\alpha} - A_{\alpha+1} + D$, and $c'_{\alpha+1} = c_{\alpha}$.
	Intuitively, this construction amounts to swapping the two strokes while routing both through the dominating set $D$.
	Both $A'_{\alpha}$ and $A'_{\alpha+1}$ are valid {\areas}, since they are connected via~$D$.

	We claim that after stroke $\alpha+1$, every vertex $v \in V(G) - D$ has the same color under $P'$ as under $P$.
	Indeed, neither of the modified strokes affects any vertex outside $A_{\alpha} + A_{\alpha+1} + D$.
	Moreover, every vertex $v \in A_{\alpha} - A_{\alpha+1} - D$ is correctly painted with color $c_{\alpha} = p_{\alpha+1}(v)$ by stroke $\alpha+1$ of $P'$, since $A'_{\alpha+1}$ contains this set.
	Similarly, every vertex $v \in A_{\alpha+1} - D$ receives color $c_{\alpha+1} = p_{\alpha+1}(v)$ under $P'$, because $A'_{\alpha} \supseteq A_{\alpha+1}$ and $A'_{\alpha+1}$ does not recolor vertices of $A_{\alpha+1} - D$.

	For vertices $v \in D$, those with $f_P(v) = \alpha$ remain colored with $c_{\alpha} = p_{\alpha+1}(v)$ under $P'$, as stroke $\alpha+1$ repaints all of $D$ with this color.
	By assumption, all remaining vertices $v \in D$ satisfy $f_P(v) > \alpha+1$ and are therefore finalized identically by $P$ and $P'$.
	Consequently, $P'$ is still an optimal {\plan} for~$t$.
	However, with respect to $P'$, the induced $P_4$ given by $D$ has finishing index $f_{P'}(D) = f_P(D) + 1$, implying $f_{P'} > f_P$.
	This contradicts the maximality of $f_P$ and completes the proof of Claim~\ref{cla:canonical_optplan:one}.

	\medskip
	To implement the same idea for the remaining indices, we again derive a contradiction to the maximality of $f_P$.
	This time, however, we can no longer use $D$ as a hub, which makes the argument more delicate.
	\begin{claim}\label{cla:canonical_optplan:two}
		For every $\sigma \in \{\beta, \gamma, \delta\}$ with $\sigma + 3 \leq s$, at least one of the strokes $\sigma+1$, $\sigma+2$, or $\sigma+3$ belongs to $\{\gamma, \delta, s\}$.
	\end{claim}
	For brevity, let $x = \sigma+1$, $y = \sigma+2$, and $z = \sigma+3$.
	In the following, we assume for all $i \in \{\sigma,x,y\}$ and all $j \in \{i+1, \dots, z\}$ that either $A_i \supseteq A_j$ or $A_i \cap A_j = \emptyset$, and we first justify that this assumption is without loss of generality.

	Indeed, if this condition were violated, we could replace $A_i$ in stroke $i$ by a set $A'_i$, defined as the smallest superset of $A_i$ that contains the entire {\area} $A_j$ for every $j \in \{i+1,\dots,z\}$ that fulfills $A_i \cap A_j \neq \emptyset$.
	Clearly, $A'_i$ is a valid {\area}, since it is the union of at most four connected sets among $A_\sigma, A_x, A_y$, and $A_z$.
	Moreover, the painting $p_z$ produced by the modified plan coincides with that of $P$.
	This holds because all paintings up to $p_{i-1}$ remain unchanged and any vertex $v \in A'_i - A_i$ that is painted prematurely in stroke $i$ is corrected by the last subsequent stroke $A_j$ with $v \in A_j$ and $j \in \{i+1,\dots,z\}$.
	As no strokes after $z$ are altered, the resulting sequence remains a valid {\plan} for $t$.
	Finally, the finishing times of all vertices remain unchanged, and hence the finishing index is preserved.
	Thus, the assumption is justified.

	Now assume for a contradiction that $x,y,z \notin \{\gamma, \delta, s\}$.
	We claim that at least one of the following three sequences -- each of which moves one of the strokes $x$, $y$, or $z$ directly in front of $\sigma$ -- defines a valid {\plan}.
	Since this postpones stroke $\sigma$ by exactly one position, the finishing index of $D$ would increase, contradicting the maximality of $f_P$.
	\begin{table}[h]
		\begin{center}	
			\begin{tabular}{c | c | c}
				Plan $P_1$ & Plan $P_2$ & Plan $P_3$ \\
				\hline
				$(A_1, c_1)$ 					& $(A_1, c_1)$ & 					$(A_1, c_1)$ \\
				\vdots 							& \vdots 							& \vdots \\
				$(A_x, c_x)$               		& $(A_y, c_y)$               		& $(A_z, c_z)$ \\
				$(A_\sigma - F_P(x), c_\sigma)$	& $(A_\sigma - F_P(y), c_\sigma)$	& $(A_\sigma - F_P(z), c_\sigma)$ \\
				$(A_y, c_y)$               		& $(A_x - F_P(y), c_x)$ 			& $(A_x - F_P(z), c_x)$ \\
				$(A_z, c_z)$               		& $(A_z, c_z)$               		& $(A_y - F_P(z), c_y)$ \\
				\vdots 							& \vdots 							& \vdots \\
				$(A_s, c_s)$ 					& $(A_s, c_s)$ 						& $(A_s, c_s)$
			\end{tabular}
		\end{center}
		\caption{Adjusted sequences for $P$.}
		\label{table:strats}
	\end{table}
	
	For each $j \in \{1,2,3\}$, the revised sequence $P_j$ moves the stroke $(A_{\sigma+j},c_{\sigma+j})$ directly in front of $(A_\sigma,c_\sigma)$, thereby delaying the latter by one position.
	This requires slight adjustments to $A_\sigma$ and the subsequent {\areas}.
	We now show that each such sequence correctly reproduces $t$, provided that the modified sets are valid {\areas}.

	Fix $j \in \{1,2,3\}$ and consider $P_j$ which, according to Table~\ref{table:strats}, is defined as the sequence $(A'_1, c'_1), \dots, (A'_s, c'_s)$ with
	\begin{itemize}
		\item $(A'_i, c'_i) = (A_i, c_i)$, for all $i \in [1,s] - [\sigma, \sigma+j]$,
		\item $(A'_\sigma, c'_\sigma) = (A_{\sigma+j}, c_{\sigma+j})$,
		\item $(A'_{\sigma+1}, c'_{\sigma+1}) = (A_\sigma-F_P(\sigma+j), c_\sigma)$, and
		\item $(A'_i, c'_i) = (A_{i-1}-F_P(\sigma+j), c_{i-1})$, for all $i \in [\sigma+2,\sigma+j]$.
	\end{itemize}
	Assuming that all $(A'_i,c'_i)$ are valid {\areas}, let $p'_0,\dots,p'_s$ denote the paintings induced by $P_j$.
	Clearly, $p'_i = p_i$ for all $i \leq \sigma-1$.
	We next show that after stroke $\sigma+j$, every vertex $v$ with $f_P(v) \leq \sigma+j$ has the correct color under $P_j$.

	Vertices outside $A_\sigma + \dots + A_{\sigma+j} + D$ are unaffected.
	For $v$ inside this union with $f_P(v) \leq \sigma+j$, we distinguish two cases:
	\begin{itemize}
    	\item If $v \in A_{\sigma+j}$, then $v \in F_P(\sigma+j)$ and therefore $p'_{\sigma+j}(v) = p_{\sigma+j}(v) = c_{\sigma+j}$, since $v \in A'_\sigma$ and none of the subsequent adjusted {\areas} of $P_j$ overwrites $F_P(\sigma+j)$.
    	\item If $v \in A_{\sigma+i} - \hat{A}_i$, where $\hat{A}_i = A_{\sigma+i+1} + \dots + A_{\sigma+j}$ for some $i \in [0, j-1]$, then $v$ is correctly colored $c_{\sigma+i}$ in $p'_{\sigma+j}$. 
		This is because the adjusted {\area} $A'_{\sigma+i+1} = A_{\sigma+i} - F_P(\sigma+j)$ is a superset of $A_{\sigma+i} - A_{\sigma+j} \supseteq A_{\sigma+i} - \hat{A}_i$ and, thus, contains $v$.
		At the same time, all subsequent adjusted {\areas}
		\begin{align*}
			A'_{\sigma+i+2} &= A_{\sigma+i+1} - F_P(\sigma+j) \subseteq A_{\sigma+i+1} \subseteq \hat{A}_i,\\
			&\ \vdots\\ 
			A'_{\sigma+j} &= A_{\sigma+j-1} - F_P(\sigma+j) \subseteq A_{\sigma+j-1} \subseteq \hat{A}_i
		\end{align*}
		are subsets of $\hat{A}_i$ and, thus, cannot contain or overwrite $v$.
	\end{itemize}
	Thus, any remaining discrepancies concern only vertices $v$ with $f_P(v) > \sigma+j$.
	Since $P$ and $P_j$ coincide thereafter, all such vertices are finalized identically, yielding $p'_s = p_s = t$.

	Finally, consider vertices $v \in D$.
	If $f_P(v) < \sigma$ or $f_P(v) > \sigma+j$, then $f_{P_j}(v) = f_P(v)$ because the strokes in $[1,\sigma-1]$ and $[\sigma+j+1,s]$ are the same for $P$ and $P_j$.
	If $f_P(v) = \sigma$, then $f_{P_j}(v) > f_P(v)$ by construction.
	The case $f_P(v) \in \{x,y,z\}$ is excluded by $\alpha \leq \beta \leq \sigma$ and the assumption on $x,y,z$.
	Hence, the finishing index of $D$ increases under $P_j$, contradicting the choice of $P$.

	\medskip
	It remains to show that at least one of $P_1,P_2,P_3$ is a valid {\plan}.
	That is, for some $j \in \{1,2,3\}$, the adjusted strokes $\sigma,\dots,\sigma+j$ must operate on valid {\areas}.
	Stroke $\sigma$ poses no issue, as it coincides with $(A_{\sigma+j},c_{\sigma+j})$ from $P$.

	We now examine the remaining adjustments.
	Since $P$ is optimal, the sets $F_P(\sigma), F_P(x), F_P(y)$, and $F_P(z)$ are all non-empty.
	Moreover, they are pairwise disjoint because every vertex can be finished only once.
	As $A_\sigma$ is either disjoint with or a superset of $A_i$, we also have $A_\sigma \cap F_P(i) = \emptyset$ or $A_\sigma \supseteq F_P(i)$ for all $i \in \{x,y,z\}$.
	Each graph $G[A_i]$ for $i \in \{\sigma,x,y,z\}$ is a connected cograph, since at least one vertex of every induced $P_4$ has already been finished prior to stroke $i$.

	\smallskip
	Consider $P_1$.
	The only modified stroke is $\sigma+1$, which uses $A_\sigma - F_P(x)$.
	This is a valid {\area} unless $F_P(x)$ separates $G[A_\sigma]$.
	If it does not, $P_1$ is a valid plan and we are done.

	\smallskip
	Otherwise, $F_P(x)$ is a separator of $G[A_\sigma]$, which implies $A_x \cap A_\sigma \not=\emptyset$ and, thus, $A_x \subseteq A_\sigma$.
	Then we turn to $P_2$.
	The set $A_\sigma - F_P(y)$ is connected, since otherwise, $F_P(x)$ and $F_P(y)$ are disjoint separators of $G[A_\sigma]$ and Observation~\ref{obs:cograph_separator_cover} would imply $F_P(x) + F_P(y) = A_\sigma$, leaving no room for the non-empty and disjoint set $F_P(\sigma)$.
	Thus, the second adjustment in $P_2$ is valid.
	If $F_P(y)$ does not separate $G[A_x]$, then $A_x - F_P(y)$ is also connected and $P_2$ is a valid plan.

	\smallskip
	If neither $P_1$ nor $P_2$ is valid, then $F_P(x)$ separates $G[A_\sigma]$ and $F_P(y)$ separates $G[A_x]$.
	This implies that $F_P(y)$ has vertices in $A_x$ and, like above, this means $A_y \subseteq A_x \subseteq A_\sigma$.
	In this case, we consider $P_3$.
	As before, $A_\sigma - F_P(z)$ is connected, since, like for $F_P(y)$, the cograph $G[A_\sigma]$ cannot be separated by $F_P(z)$.
	Similarly, $F_P(z)$ cannot separate $G[A_x]$, as this would imply that $G[A_x]$ has two disjoint separators, $F_P(y)$ and $F_P(z)$, which means $F_P(y) + F_P(z) = A_x$ due to Observation~\ref{obs:cograph_separator_cover}.
	Again, this contradicts the existence of the non-empty $F_P(x) \subseteq A_x$.
	Thus, $A_x - F_P(z)$ is connected.
	
	Finally, consider $A_y - F_P(z)$.
	This set is a valid {\area} unless $F_P(z)$ separates $G[A_y]$.
	Assume for a contradiction that this is the case and that $G[A_y] - F_P(z)$ is disconnected.
	Then, as a separator of $G[A_y]$, $F_P(z)$ has vertices in $A_y$ and, again, we get $A_z \subseteq A_y \subseteq A_x \subseteq A_\sigma$.

	\smallskip
	To derive the contradiction, we begin by selecting a vertex $u \in F_P(\sigma)$, which is not in $F_P(x)$, and a vertex $u' \in A_\sigma$ that is in a different component of $G[A_\sigma] - F_P(x)$ than $u$; clearly, $uu' \not\in E(G)$.
	By Observation~\ref{obs:cograph_separator_pthree} there is a vertex $v$ in the separator $F_P(x)$ that is adjacent to both $u$ and $u'$.
	
	Next, there is a vertex $v' \in F_P(x)$ that is in a different component $C$ of $G[A_x] - F_P(y)$ than $v$.	
	To see this, observe first that $F_P(x)$ dominates $A_\sigma$ due to Observation~\ref{obs:cograph_separator_dominate}.
	With $C \subseteq A_x \subseteq A_\sigma$ and this implies that $F_P(x)$ dominates $C$, as well.
	Moreover, as $F_P(x)$ and $F_P(y)$ are disjoint, it follows that $F_P(x) \subseteq A_x - F_P(y)$.
	Thus, being a dominating set, $F_P(x)$ must have at least one vertex $v'$ in $C$ (which is not in the same connected component with $v$, by design).

	Using Observation~\ref{obs:cograph_separator_pthree}, again, there is a vertex $w \in F_P(y)$ adjacent to $v$ and $v'$.
	Finally, Observation~\ref{obs:cograph_separator_dominate} implies that $F_P(z)$ dominates $A_y$ and, thus, there is a vertex $m \in F_P(z)$ that is adjacent to $w$.

	Observe that, due to the properties of separators and the disjointness of $F_P(\sigma), F_P(x), F_P(y)$, and $F_P(z)$, all vertices of $u,v,v',w,m$ are distinct.
	For $u'$, we are only guaranteed that it is not the same as $u$ and $v$.
	We show that it is valid to assume that $u$ cannot see $w$ and $m$.
	In fact, we show, if at least one of the two edges is present then $u'$ is distinct from and not adjacent to $w$ and $m$.
	Hence, we could safely replace $u$ with $u'$.
	So, if $uw \in E(G)$ then $u' \not= w$, as $uu' \not\in E(G)$.
	Furthermore, the presence of $uw$ means that $u' \not= m$ since otherwise $u,w,m$ would be a path connecting $u$ and $u'$ in $G[A_\sigma]-F_P(x)$.
	Clearly, if $u'$ sees $w$ or $m$ we get a similar connecting path $u,w,u'$ or $u,w,m,u'$.
	The same argument works, if $u$ sees $m$, simply by swapping $w$ and $m$ in the connecting paths.

	Since $G[A_\sigma]$ is $P_4$-free, the path $u,v,w,v'$ cannot be induced and, since $uw,vv' \not\in E(G)$, the edge $uv'$ must be present.
	Similarly, the path $u,v,w,m$ cannot be induced and the absence of $uw$, $um$ implies $vm \in E(G)$.
	But then, $v'm$ is absent since, otherwise, the path $v,m,v'$ connects $v$ and $v'$ in $G[A_x] - F_P(y)$.
	But then the path $u,v',w,m$ is induced.
	This contradiction proves Claim~\ref{cla:canonical_optplan:two}.

	\medskip
	Using the claims, we obviously have $s - \alpha \leq 12$.
	This completes the proof.
\end{proof}

\LemHeads*
\begin{proof}
	We initialize $H_s \gets \emptyset$ and enumerate all subsets $D \subseteq V(G)$ with $|D|=4$.
	For each such set, we discard those $D$ that do not induce a $P_4$ in~$G$.
	For every remaining choice of~$D$, we enumerate all subsets $F \subseteq C$ with $|F| = s-12$.

	For each pair $(D,F)$, we construct the partial {\plan}
	\[
		H_{D,F} = (t^{-1}(c_1)+ D, c_1),\dots,(t^{-1}(c_{s-12})+ D, c_{s-12}),
	\]
	where $c_1,\dots,c_{s-12}$ are the elements of~$F$ ordered according to~$\prec$,
	and add $H_{D,F}$ to~$H_s$.
	Clearly, $H_{D,F}$ is a $(12,\prec)$-canonical {\plan} missing only its final $12$ strokes.
	Hence, it is indeed the head of some complete $(12,\prec)$-canonical {\plan}.

	We now show that $H_s$ contains all required heads.
	Let $P$ be any $(12,\prec)$-canonical {\plan} of length~$s$ with dominating set~$D$, and let $H$ denote its head.
	Since we enumerate all $4$-vertex subsets, the set~$D$ is considered.
	Moreover, we also enumerate exactly the set~$F$ of colors used by~$H$.
	By definition of a canonical head, the strokes of~$H$ are precisely of the form $(t^{-1}(c)+ D,c)$ for $c\in F$.
	Because these strokes are ordered according to~$\prec$, reconstructing them by iterating over~$F$ in this order reproduces~$H$ exactly.
	Thus, $H\in H_s$.

	Finally, the construction runs in polynomial time.
	Enumerating all choices of $D$ and $F$ requires only polynomial time.
	Checking whether a $4$-vertex set induces a $P_4$, sorting~$F$, and constructing the corresponding strokes can all be done efficiently.
	Since at most one head is added for each pair $(D,F)$, maintaining the set~$H_s$ also takes polynomial time.
\end{proof}

\ObsMaximal*
\begin{proof}
	Let $P$ be a $(D,k,\prec)$-canonical {\plan} of length~$s$ that is not maximal.
	Then there exists an index $i \in [1,s]$ and a vertex $v\in V(G)$ such that replacing the area~$A_i$ of the $i$th stroke by $A_i+ v$ yields another valid $(D,k,\prec)$-canonical {\plan}.

	Since the graph has finitely many vertices and strokes, this augmentation process can only be applied finitely many times.
	Once no further vertex can be added to any stroke without violating canonicity, the resulting {\plan} is maximal by definition.
\end{proof}

\LemTails*
\begin{proof}
	We initialize $T_s \gets \emptyset$ and enumerate all sequences $F=(c_1,\dots,c_{12})\in C^{12}$.
	For each such sequence, we enumerate all vertex sequences $V=(v_1,\dots,v_{12})\in V(G)^{12}$.
	From these vertices, we compute sets $V_i$ for $i \in [1,12]$ by collecting all vertices reachable from $v_i$ in the graph
	\[
		G\bigl[t^{-1}(c_i)+ \bigcup_{j=i+1}^{12} V_j\bigr].
	\]
	Intuitively, this consists of all vertices that either must receive color~$c_i$ in the final template or will be recolored later.
	We then add the partial {\plan}
	\[
		T_{F,V} = (V_1,c_1),\dots,(V_{12},c_{12})
	\]
	to~$T_s$.

	To see that $T_s$ contains all desired tails, let $P$ be any maximal $(12,\prec)$-canonical {\plan} of length~$s$, and let
	\[
		T=(A_1,c_1),\dots,(A_{12},c_{12})
	\]
	be its tail.
	Maximality implies that no vertex can be added to any~$A_i$ while preserving canonicity.
	Eventually, the brute-force enumeration produces the color sequence $(c_1,\dots,c_{12})$.
	Moreover, we enumerate a vertex sequence $(w_1,\dots,w_{12})$ with $w_i\in A_i$ for all~$i$.
	Let $T_{F,V}=(A_1',c_1),\dots,(A_{12}',c_{12})$ be the tail constructed from this choice.

	By construction, $w_i\in A_i\cap A_i'$ for all~$i$.
	Assume for contradiction that $A_j\neq A_j'$ for some maximal index~$j$.
	Then $A_k=A_k'$ for all $k>j$.
	For any $v\in A_j$, validity of~$P$ implies that either $t(v)=c_j$ or $v\in A_k$ for some $k>j$.
	By construction of~$A_j'$, this implies $v\in A_j'$, and hence $A_j\subseteq A_j'$.

	Conversely, let $v'\in A_j'$ such that $N(v')$ intersects~$A_j$.
	Then adding $v'$ to~$A_j$ preserves canonicity:
	either $t(v')=c_j$, or $v'$ will be recolored later.
	This contradicts maximality of~$P$.
	Hence, no such index~$j$ exists and $A_i=A_i'$ for all~$i$.
	Thus, $T$ is reproduced correctly.

	The total enumeration of $C^{12}\times V(G)^{12}$ is polynomial.
	Each area $A_i'$ is computed using a bounded number of breadth-first searches on a modified graph, which is polynomial.
	Since only polynomially many tails are generated and stored, $T_s$ can be constructed in polynomial time.
\end{proof}

\LemHeadsAndTails*
\begin{proof}
	By Lemmas~\ref{lem:heads} and~\ref{lem:tails}, the sets $H_s$ and $T_s$ can be generated in polynomial time and have polynomial size.
	We therefore enumerate all pairs $(h,t)\in H_s\times T_s$.

	For each such pair, we verify in polynomial time whether appending~$t$ to~$h$ yields a valid {\plan}.
	All valid combinations are collected into~$P_s$, which can be maintained efficiently since it contains only polynomially many elements.

	It remains to show that every {\plan} in~$P_s$ is maximal.
	Let $P=(A_1,c_1),\dots,(A_s,c_s)\in P_s$ and assume for contradiction that $P$ is not maximal.
	Then there exists a vertex $v\in V(G)$ and an index $i \in [1,s]$ such that adding $v$ to~$A_i$ preserves canonicity.

	Since no head stroke can be extended, we must have $i>s-12$.
	Validity of the extension implies that $N(v)$ intersects~$A_i$ and that either $t(v)=c_i$ or $v\in A_j$ for some $j>i$.
	Hence, $v$ lies in the graph used to construct~$A_i$, and therefore $v\in A_i$ by construction -- a contradiction.

	Thus, all {\plans} in~$P_s$ are maximal $(12,\prec)$-canonical {\plans} of length~$s$.
\end{proof}

\ThmMainResult*
\begin{proof}
	The proof follows from the correctness and efficiency of Algorithm~\ref{alg:co-gem}.

	Let $(G,t)$ be an instance with an optimal {\plan} of length~$s$, and let $\prec$ be an arbitrary ordering of the colors used by~$t$.
	If $G$ is a cograph, correctness follows from Theorem~\ref{cor:painting_cographs}.
	Otherwise, Section~\ref{sec:canonical_paint_plans} shows that $c\leq s\leq c+3$.

	By Lemma~\ref{lem:canonical_optplan} and Observation~\ref{obs:maximal}, there exists a maximal $(D,12,\prec)$-canonical {\plan} of length~$s$ for some induced $P_4$~$D$.
	Thus, Lemma~\ref{lem:heads_and_tails} guarantees that $P_s\neq\emptyset$.
	Conversely, no canonical {\plan} of length $s-1$ exists, so $P_{s-1}=\emptyset$.

	Hence, the first iteration of the loop that encounters a non-empty set is exactly the iteration with $i=s$, and the algorithm returns an optimal {\plan}.

	Polynomial running time follows immediately.
	Cograph recognition is polynomial~\cite{Corneil1985} as is solving according to Theorem~\ref{cor:painting_cographs}.
	The loop has at most four iterations, and Lemma~\ref{lem:heads_and_tails} guarantees polynomial-time generation of each~$P_s$.
\end{proof}

%

\section*{Disclaimer on the Use of AI}

We note that the AI system \textsc{ChatGPT}, in particular model \emph{GPT-5 mini}, has been used in this work for the only purpose of improving the presentation and readability of the text, as the authors are not native English speakers.
No other use of \textsc{ChatGPT} or its subroutines has been made in this paper.
In particular, no AI framework has contributed to any creative decisions, the development of original results, or the construction of rigorous mathematical proofs, all of which were carried out entirely by the authors.

%

\bibliographystyle{abbrv}
\bibliography{floodit.bib}

\end{document}